\documentclass{article}

\usepackage{amsthm, amssymb, amsmath}
\usepackage{verbatim}
\usepackage{graphics}
\usepackage{graphicx}
\usepackage{setspace}

\newtheorem{thm}{Theorem}
\newtheorem{prop}[thm]{Proposition}
\newtheorem{cor}[thm]{Corollary}
\newtheorem{lem}[thm]{Lemma}

\theoremstyle{plain}
\newtheorem*{defn}{Definition}

\newcommand{\R}{\mathbb{R}}

\newcommand{\var}{\textrm{Var}}

\newcommand{\pr}{\textrm{Pr}}
\newcommand{\sgn}{\textrm{sgn}}
\newcommand{\E}{\textrm{E}}

\title{A Small PRG for Polynomial Threshold Functions of Gaussians}
\author{Daniel M. Kane}

\begin{document}

\maketitle

\section{Introduction}

A polynomial threshold function (PTF) is a function of the form $f(X)=\sgn(p(X))$ for some polynomial $p(X)$.  We say that $f$ is a degree-$d$ polynomial threshold function of $p$ is of degree at most $d$.  Polynomial threshold functions are a fundamental class of functions with applications to many fields such as circuit complexity \cite{curciutApp}, communication complexity \cite{commApp} and learning theory \cite{learningApp}.

We discuss the issue of pseudo-random generators for polynomial threshold functions of bounded degree.  Namely for some known probability distribution $\mathcal{D}$ on $\R^n$, we would like to find an explicit, easily computable function $G:\{0,1\}^S \rightarrow \R^n$ so that for any degree-$d$ polynomial threshold function, $f$,
$$
\left|\E_{Y\sim \mathcal{D}}[f(Y)] - \E_{X\in_u \{0,1\}^S}[f(G(X))] \right| < \epsilon.
$$
There are two natural distributions, $\mathcal{D}$, to study for this problem.  The first is that of the hypercube distribution, namely the uniform distribution over $\{0,1\}^n$.  The second is the Gaussian distribution.  The latter can often be thought of as a special case of the former.  In particular for polynomials of low influence (for which no one variable has significant control over the size of the polynomial), the invariance principle says that these polynomials behave similarly on the two distributions.  In fact many results about the hypercube distribution are proven by using the invariance principle to reduce to the Gaussian case where symmetry and the continuous nature of the random variables make things considerably easier.

In this paper we construct an explicit PRG for the Gaussian case.  In particular, for any real numbers $c,\epsilon>0$ and integer $d>0$ we construct a PRG fooling degree-$d$ PTFs of Gaussians to within $\epsilon$ of seed length $\log(n)2^{O_c(d)}\epsilon^{-4-c}.$  In particular we show that

\begin{thm}\label{MainThm}
Let $c>0$ be a constant.  For $\epsilon>0$, and $d$ a positive integer, let $N=2^{\Omega_c(d)}\epsilon^{-4-c}$ and $k= \Omega_c(d)$ be integers.  Let $X_i$ $1\leq i \leq N$ be independently sampled from $k$-independent families of standard Gaussians.  Let $X=\frac{1}{\sqrt{N}}\sum_{i=1}^N X_i$.  Let $Y$ be a fully independent family of standard Gaussians.  Then for any degree-$d$ polynomial threshold function, $f$,
$$
|\E[f(X)] - \E[f(Y)]| < \epsilon.
$$
\end{thm}

From this we can construct an efficient PRG for PTFs of Gaussians.  In particular

\begin{cor}\label{PRGCor}
For every $c>0$, there exists a PRG that $\epsilon$-fools degree-$d$ PTFs of Gaussians with seed length
$$
\log(n)2^{O_c(d)}\epsilon^{-4-c}.
$$
\end{cor}

Much of the previous work in constructing pseudo-random generators involves the use of functions of limited independence.  It was shown in \cite{DGJSV} that $\tilde{O}(\epsilon^{-2})$-independence fools degree-$1$ PTFs.  The degree-$2$ case was later dealt with in \cite{DKN}, in which it was shown that $\tilde{O}(\epsilon^{-9})$-independence sufficed (and that $O(\epsilon^{-8})$ sufficed for Gaussians).  The author showed that limited independence suffices to fool arbitrary degree PTFs of Gaussians in \cite{K}, but the amount of independence required was $O_d(\epsilon^{-2^{O(d)}})$.  In terms of PRGs that do not rely solely on limited independence, \cite{MZ} found a PRG for degree-$d$ PTFs on the hypercube distribution of size $\log(n) 2^{O(d)} \epsilon^{-8d-3}$.  Hence for $d$ more than constant sized, and $\epsilon$ less than some constant, our PRG will always beat the other known examples.

The basic idea of the proof of Theorem \ref{MainThm} will be to show that $X$ fools a function $g$ which is a smooth approximation of $f$.  This is done using the replacement method.   In particular, we replace the $X_i$ by fully independent families of Gaussians one at a time and show that at each step a small error is introduced.  This is done by replacing $g$ by its Taylor expansion and noting that small degree moments of the $X_i$ are identical to the corresponding moments of a fully independent family.

Naively, if $f=\sgn(p(x))$, we might try to let $g=\rho(p(x))$ for some smooth function $\rho$ so that $\rho(x)=\sgn(x)$ for $|x|>\delta$.  If we Taylor expand $g$ to order $T-1$, we find that the error in replacing $X_i$ by a fully random Gaussian is roughly the size of the $T^{th}$ derivative of $g$ times the $T^{th}$ moment of $p(X)-p(X')$, where $X'$ is the new random variable we get after replacing $X_i$.  We expect the former to be roughly $\delta^{-T}$ and the latter to be roughly $|p|_2^T N^{-T/2}$.  Hence, for this to work we will need $N\gg (|p|_2 \delta^{-1})^2$.  On the other hand, for $g$ to be a good approximation of $f$, we will need that the probability that $|p(Y)|<\delta$ to be small.  Using standard anti-concentration bounds, this requires $|p|_2\delta^{-1}$ to be roughly $\epsilon^{-d}$, and hence $N$ will be required to be at least $\epsilon^{-2d}$.

In order to fix this, we use a better notion of anti-concentration.  Our underlying heuristic is that for any polynomial $p$ it should be the case that $|p(X)|<\epsilon|p'(X)|$ with probability not much bigger than $\epsilon$.  This should hold because changing the value of $X$ by $\epsilon$ should adjust the value of $p(X)$ by roughly $\epsilon |p'(X)|$.  This allows us to state a strong version of anti-concentration.  In particular with probability roughly $1-\epsilon$ it should hold that
\begin{equation}\label{StrongAntiConcentrationEquation}
|p(X)| \geq \epsilon|p'(X)| \geq \epsilon^2|p''(X)| \geq \ldots \geq \epsilon^d |p^{(d)}(X)|.
\end{equation}
So $|p(X)|\geq \epsilon^d |p^{(d)}(X)|.$  It should be noted that $|p^{(d)}(X)|$ is independent of $X$ and can be thought of as a rough approximation to $|p|_2$.  In our analysis we will use a $g$ that is a function not only of $p(X)$, but also of $|p^{(m)}(X)|$ for $1\leq m \leq d$.  Instead of forcing $g$ to be a good approximation to $f$ whenever $|p(X)|\geq \epsilon^d$, we will only require it to be a good approximation to $f$ at $X$ where Equation \ref{StrongAntiConcentrationEquation} holds.  This gives us significant leeway since although the derivative to $g$ with respect to $p(X)$ will still be large at places, this will only happen when $|p'(X)|$ is small, and this in turn will imply that the variance in $p(X)$ achieved by replacing $X_i$ is comparably small.

In Section \ref{basicStuffSec}, we will review some basic properties of polynomial threshold functions.  In Section \ref{DerivativeSec}, we in introduce the notion of the derivative (which we call the noisy derivative) that will be useful for our purposes.  We then prove a number of Lemmas about this derivative and in particular prove a rigorous version of Equation \ref{StrongAntiConcentrationEquation}.  In Section \ref{averagingSec}, we discuss some averaging operators that will be useful in analyzing what happens when one of the $X_i$ is changed.  In Section \ref{mainSec}, we use these results and the above ideas to prove Theorem \ref{MainThm}.  In Section \ref{FiniteEntropySection}, we use this result to prove Corollary \ref{PRGCor}.

\section{Definitions and Basic Properties}\label{basicStuffSec}

We are concerned with polynomial threshold functions, so for completeness we give a definition
\begin{defn}
A degree-$d$ polynomial threshold function (PTF) is a function of the form
$$
f = \sgn(p(x))
$$
where $p$ is a polynomial of degree at most $d$.
\end{defn}

We are also concerned with the idea of fooling functions so we define
\begin{defn}
Let $f:\R^n\rightarrow \R$ we say that a random variable $X$ with values in $\R^n$ $\epsilon$-fools $f$ if
$$
|\E[f(X)] - \E[f(Y)]| \leq \epsilon
$$
where $Y$ is a standard $n$-dimensional Gaussian.
\end{defn}

For convenience we define the notation:
\begin{defn}
We use
$$
A \approx_\epsilon B
$$
to denote that
$$
|A-B| = O(\epsilon).
$$
\end{defn}

For a function on $\R^n$ we define its $L^k$ norm by
\begin{defn}
For $p:\R^n\rightarrow \R$ define
$$
|p|_k = \E_X [|p(X)|^k]^{1/k}.
$$
Where the above expectation is over $X$ a standard $n$-dimensional Gaussian.
\end{defn}

We will make use of the hypercontractive inequality.  The proof follows from Theorem 2 of \cite{hypercontractivity}.

\begin{lem}\label{hypercontractiveLem}
If $p$ is a degree-$d$ polynomial and $t>2$, then
$$
|p|_t \leq \sqrt{t-1}^d |p|_2.
$$
\end{lem}

In particular this implies the following Corollary:

\begin{cor}\label{NotTooSmallCor}
Let $p$ be a degree-$d$ polynomial in $n$ variables.  Let $X$ be a family of standard Gaussians.  Then
$$
\pr\left(|p(X)|\geq |p|_2/2\right) \geq 9^{-d}/2.
$$
\end{cor}
\begin{proof}
This follows immediately from the Paley–Zygmund inequality applied to $p^2$.
\end{proof}

We also obtain:
\begin{cor}\label{GeneralHypercontractiveCor}
Let $p$ be a degree-$d$ polynomial, $t\geq 1$ a real number, then
$$
|p|_t \leq 2^{O_t(d)} |p|_1
$$
\end{cor}
\begin{proof}
By Lemma \ref{hypercontractiveLem}, it suffices to prove this for $t=2$.  This in turn follows from Corollary \ref{NotTooSmallCor}, which implies that $|p|_1 \geq 9^{-d}/4 |p|_2$.
\end{proof}

And
\begin{cor}\label{AverageInequalityCor}
If $p(X,Y),q(X,Y)$ are degree-$d$ polynomials in standard Gaussians $X$ and $Y$ then
$$
\pr_Y (|p(X,Y)|_{2,X} < \epsilon |q(X,Y)|_{2,X}) \leq 4\cdot 9^d \pr_{X,Y}(|p(X,Y)| < 4\cdot 3^d \epsilon |q(X,Y)|).
$$
Where $|r(X,Y)|_{2,X}$ denotes the $L^2$ norm over $X$, namely $\left(\E_X[r(X,Y)^2]\right)^{1/2}$.
\end{cor}
\begin{proof}
Given $Y$ so that $|p(X,Y)|_{2,X} < \epsilon |q(X,Y)|_{2,X}$, by Corollary \ref{NotTooSmallCor} we have that $|q(X,Y)|\geq |q(X,Y)|_{2,X}/2$ with probability at least $9^{-d}/2$.  Furthermore, $|p(X,Y)|\leq2\cdot 3^d |p(X,Y)|_{2,X}$ with probability at least $1-9^{-d}/4$.  Hence with probability at least $9^{-d}/4$ we have that
$$
|p(X,Y)| \leq 2\cdot 3^d |p(X,Y)|_{2,X} < 2\cdot 3^d \epsilon |q(X,Y)|_{2,X} \leq 4\cdot 3^d \epsilon |q(X,Y)|.
$$
So
\begin{align*}
\pr_{X,Y}(|p(X,Y)| < 4\cdot 3^d \epsilon |q(X,Y)|) \geq \frac{9^{-d}}{4} \pr_Y (|p(X,Y)|_{2,X} < \epsilon |q(X,Y)|_{2,X}).
\end{align*}
\end{proof}

\section{Noisy Derivatives}\label{DerivativeSec}

In this section we define our notion of the noisy derivative and obtain some of its basic properties.

\begin{defn}
Let $X$ and $Y$ be $n$-dimensional vectors and $\theta$ a real number.  Then we let
$$
N_Y^\theta(X) := \cos(\theta) X + \sin(\theta) Y.
$$
\end{defn}
If $X$ and $Y$ are independent Gaussians, $Y$ can be thought of as a noisy version of $X$ with $\theta$ a noise parameter.  We next define the noisy derivative.
\begin{defn}
Let $X,Y,Z$ be $n$-dimensional vectors, $f:\R^n\rightarrow \R$ a function, and $\theta$ a real number.  We define the \emph{noisy derivative of $f$ at $X$ with parameter $\theta$ in directions $Y$ and $Z$} to be
$$
D_{Y,Z}^\theta f(X) := \frac{f(N_Y^\theta(X)) - f(N_Z^\theta(X))}{\theta}.
$$
\end{defn}
It should be noted that if $X,Y,Z$ are held constant and $\theta$ goes to 0, the noisy derivative approaches the difference of the directional derivatives of $f$ at $X$ in the directions $Y$ and $Z$.  The noisy derivative for positive $\theta$ can be thought of as sort of a large scale derivative that covers slightly more than just a differential distance.

We also require a notion of the average size of the $\ell^{th}$ derivative.  In particular we define:
\begin{defn}
For $X$ an $n$-dimensional vector, $f:\R^n\rightarrow \R$ a function, $\ell$ a non-negative integer, and $\theta$ a real number, we define
$$
|f_\theta^{(\ell)}(X)|_2^2 := \E_{Y_1,Z_1,\ldots,Y_\ell,Z_\ell}[|D^\theta_{Y_1,Z_1}D^\theta_{Y_2,Z_2}\cdots D^\theta_{Y_\ell,Z_\ell}f(X)|^2],
$$
where the expectation is taken over independent, standard Gaussians $Y_i,Z_i$.
\end{defn}

\begin{lem}\label{HighDerConstLem}
For $p$ a degree-$d$ polynomial, and $\theta$ a real number
$$
|p^{(d)}_\theta(X)|_2
$$
is independent of $X$.
\end{lem}
\begin{proof}
This follows from the fact that for fixed $Y_i,Z_i$ that $$D^\theta_{Y_1,Z_1}D^\theta_{Y_2,Z_2}\cdots D^\theta_{Y_\ell,Z_\ell}p(X)$$ is independent of $X$.  This in turn follows from the fact that for any degree-$d$ polynomial $q$ and any $Y$,$Z$, $D^\theta_{Y,Z}q(X)$ is a degree-$(d-1)$ polynomial in $X$.
\end{proof}

We now prove our version of the statement that the value of a polynomial is probably not too much smaller than its derivative.
\begin{prop}\label{SizeVsDerProp}
Let $\epsilon,\theta>0$ be real numbers with $\theta = O(\epsilon)$.  Let $p$ be a degree-$d$ polynomial and let $X,Y,Z$ be standard independent Gaussians.  Then
$$
\pr_{X,Y,Z}(|p(X)| < \epsilon |D^\theta_{Y,Z}p(X)|) = O(d^2 \epsilon).
$$
\end{prop}

To prove this we use the following Lemma:
\begin{lem}\label{twidleXLem}
Let $\epsilon,\theta,p,d,X,Y$ be as above.  Then
$$
\pr_{X,Y}\left( |p(X)| < \frac{\epsilon}{\theta} |p(X) - p(N^\theta_Y(X))|\right) = O(d^2 \epsilon).
$$
\end{lem}
\begin{proof}
The basic idea of the proof will be the averaging argument discussed in the Introduction to this Part.  We note that the above probability should be the same for any independent Gaussians, $X$ and $Y$.  We let $X_\phi := N^\phi_Y(X)$.  Note that $X_\phi$ and $X_{\phi+\pi/2}$ are independent of each other.  Furthermore, $N^\theta_{X_{\phi+\pi/2}}(X_\phi) = X_{\phi+\theta}$.  Hence for any $\phi$, the above probability equals
$$
\pr_{X,Y} \left( |p(X_\phi)| < \frac{\epsilon}{\theta}|p(X_\phi) - p(X_{\phi+\theta})|\right).
$$
We claim that for any values of $X$ and $Y$, that the average value over $\phi\in [0,2\pi]$ of the above is $O(d^2 \epsilon)$.

Notice that with $X$ and $Y$ fixed, $p(X_\phi)$ is a degree-$d$ polynomial in $\cos(\phi)$ and $\sin(\phi)$.  Letting $z=e^{i\phi}$ we have that $\cos(\phi)=\frac{z+z^{-1}}{2},\sin(\phi)=\frac{z-z^{-1}}{2i}$.  Hence $p(X_\phi) = z^{-d} q(z)$ for some polynomial $q$ of degree at most $2d$.

We wish to bound the probability that
$$
\frac{\theta}{\epsilon} < \frac{|z^{-d}q(z) - z^{-d}e^{-id\theta}q(ze^{i\theta})|}{|z^{-d}q(z)|} = \frac{|q(z) - e^{-id\theta}q(ze^{i\theta})|}{|q(z)|}.
$$
For $\frac{\theta}{\epsilon}$ sufficiently small, we may instead bound the probability that
$$
\left| \log\left( \frac{e^{-id\theta}q(ze^{i\theta})}{q(z)}\right) \right| > \frac{\theta}{2\epsilon}.
$$
On the other hand, we may factor $q(z)$ as $a\prod_{i=1}^{2d}(z-r_i)$ where $r_i$ are the roots of $q$.  The left hand side of the above is then at most
\begin{align*}
d\theta + \sum_{i=1}^{2d} \left|\log\left(\frac{ze^{i\theta}-r_i}{z-r_i} \right)\right| & \leq d\theta + \sum_{i=1}^d O\left( \left|\frac{ze^{i\theta}-z}{z-r_i}\right|\right)\\
& \leq d\theta + \theta \sum_{i=1}^{2d} O\left( \frac{1}{|z-r_i|}\right).
\end{align*}

Hence it suffices to bound the probability that
$$
d+\sum_{i=1}^{2d} O\left( \frac{1}{|z-r_i|}\right) > \frac{1}{2\epsilon}.
$$
If $4\epsilon > d^{-1}$, there is nothing to prove.  Otherwise, the above holds only if
$$
\sum_{i=1}^{2d} O\left( \frac{1}{|z-r_i|}\right) > \frac{1}{4\epsilon}.
$$
This in turn only occurs when $z$ is within $O(d \epsilon)$ of some $r_i$.  For each $r_i$ this happens with probability $O(d\epsilon)$ over $\phi$, and hence by the union bound, the above holds with probability $O(d^2 \epsilon).$
\end{proof}
Note that a tighter analysis could be used to prove the bound $O(d\log(d) \epsilon)$, but we will not need this stronger result.

Proposition \ref{SizeVsDerProp} now follows immediately by noting that $|p(X)|$ is less than $\epsilon |D^\theta_{Y,Z}p(X)|$ only when either $ |p(X)|/2 < \frac{\epsilon}{\theta} |p(X) - p(N^\theta_Y(X))|$ or $ |p(X)|/2 < \frac{\epsilon}{\theta} |p(X) - p(N^\theta_Z(X))|$.  This allows us to prove our version of Equation \ref{StrongAntiConcentrationEquation}.

\begin{cor}\label{AvDerIneqCor}
For $p$ a degree-$d$ polynomial, $X$ a standard Gaussian, $\epsilon,\theta>0$ with $\theta=O(\epsilon)$, and $\ell$ a non-negative integer,
$$
\pr_X(|p^{(\ell)}_\theta(X)|_2 \leq \epsilon |p^{(\ell+1)}_\theta(X)|_2) \leq 2^{O(d)}\epsilon.
$$
\end{cor}
\begin{proof}
Let $Y_i,Z_i$ be standard Gaussians independent of each other and of $X$ for $1\leq i \leq \ell+1$.  Applying Proposition \ref{SizeVsDerProp} to $D^\theta_{Y_2,Z_2}\cdots D^\theta_{Y_{\ell+1},Z_{\ell+1}}p(X)$, we find that
\begin{align*}
\pr_{X,Y_1,Z_1} & \left( | D^\theta_{Y_2,Z_2}\cdots D^\theta_{Y_{\ell+1},Z_{\ell+1}} p(X)| \leq 4 \cdot 3^d \epsilon | D^\theta_{Y_1,Z_1}\cdots D^\theta_{Y_{\ell+1},Z_{\ell+1}} p(X)|\right)\\ & = O(d^2 3^d \epsilon).
\end{align*}
Noting that
$$
|p^{(\ell)}_\theta(X)|_2 = |D^\theta_{Y_2,Z_2}\cdots D^\theta_{Y_{\ell+1},Z_{\ell+1}} p(X)|_{2,(Y_1,Z_1,\ldots,Y_{\ell+1},Z_{\ell+1})}
$$
and
$$
|p^{(\ell+1)}_\theta(X)|_2 = |D^\theta_{Y_1,Z_1}\cdots D^\theta_{Y_{\ell+1},Z_{\ell+1}} p(X)|_{2,(Y_1,Z_1,\ldots,Y_{\ell+1},Z_{\ell+1})},
$$
Corollary \ref{AverageInequalityCor} tells us that
$$
\pr_X\left(|p^{(\ell)}_\theta(X)|_2 < \epsilon |p^{(\ell+1)}_\theta(X)|_2\right) = O(d^2 27^d \epsilon).
$$
\end{proof}

\section{Averaging Operators}\label{averagingSec}

A key ingredient of our proof will be to show that if we replace $X$ by $X'$ by replacing one of the $X_i$ by a random Gaussian, that the variance of $|p^{(\ell)}(X')|$ is bounded in terms of $|p^{(\ell+1)}(X)|$ (see Proposition \ref{VarianceBoundProp}).  Unfortunately, for our argument to work nicely we would want it bounded in terms of the expectation of $|p^{(\ell+1)}(X')|$.  In order to deal with this issue, we will need to study the behavior of the expectation of $q(X')$ for polynomials $q$, and in particular when it is close to $q(X)$. To get started on this project, we define the following averaging operator:
\begin{defn}
Let $X$ be an $n$-dimensional vector, $f:\R^n\rightarrow \R$ a function, and $\theta$ a real number.  Define
$$
A^\theta f(X) := \E_Y\left[ f(N^\theta_Y(X))\right].
$$
Where the expectation is over $Y$ a standard Gaussian.
\end{defn}
Note that the $A^\theta$ form the Ornstein-Uhlenbeck semigroup with $A^\theta=T_t$ where $\cos(\theta)=e^{-t}$.  The composition law becomes $T_{t_1}T_{t_2}=T_{t_1+t_2}$. We express this operator in terms of $\theta$ rather than $t$, since it fits in with our $N^\theta$ notation, which makes it more convenient for our purposes.

We also define averaged versions of our derivatives
\begin{defn}
For $p$ a polynomial, $\ell,m$ non-negative integers, $X$ a vector, and $\theta$ a real number let
$$
|p^{(\ell),m}_\theta (X)|_2^2:=E_{Y_1,\ldots,Y_m}\left[|p^{(\ell)}_\theta (N^\theta_{Y_1}\cdots N^\theta_{Y_m}X)|_2^2\right] = (A^\theta)^m |p^{(\ell)}_\theta (X)|_2^2.
$$
\end{defn}

We claim that for $X$ a standard Gaussian, that with fairly high probability $|p^{(\ell),m}_\theta (X)|$ is close to $|p^{(\ell)}_\theta (X)|$.  In particular:
\begin{lem}\label{AverageCloseLem}
If $p$ is a degree-$d$ polynomial, and $\ell$ and $m$ are non-negative integers, $X$ a standard Gaussian, and $\epsilon,\theta>0$, then
$$
\pr_X\left(||p^{(\ell),m+1}_\theta (X)|_2^2-|p^{(\ell),m}_\theta (X)|_2^2| > \epsilon |p^{(\ell),m}_\theta (X)|_2^2 \right) \leq 2^{O(d)} \theta\epsilon^{-1}.
$$
\end{lem}
\begin{proof}
If $||p^{(\ell),m+1}_\theta (X)|_2^2-|p^{(\ell),m}_\theta (X)|_2^2| > \epsilon |p^{(\ell),m}_\theta (X)|_2^2$, then by Corollary \ref{NotTooSmallCor} with probability $2^{O(d)}$ over a standard normal $Y$ we have that $$||p^{(\ell),m}_\theta (N^\theta_Y(X))|_2^2-|p^{(\ell),m}_\theta (X)|_2^2| > \epsilon |p^{(\ell),m}_\theta (X)|_2^2.$$  By Lemma \ref{twidleXLem}, this happens with probability $O(d^2 \theta \epsilon^{-1})$, and hence our original event happens with probability at most $2^{O(d)} \theta \epsilon^{-1}$.
\end{proof}

We bound the variance of these polynomials as $X$ changes.

\begin{prop}\label{VarianceBoundProp}
Let $p$ be a degree-$d$ polynomial, and $\theta>0$, $\ell,m$ non-negative integers, then for $X$ any vector and $Y$ a standard Gaussian
\begin{align*}
\var_Y \left[ |p^{(\ell),m}_\theta(N^\theta_Y X)|_2^2 \right]  \leq 2^{O(d)}\theta |p^{(\ell+1),m}_\theta(X)|_2^2 |p^{(\ell),m+1}_\theta(X)|_2^2.
\end{align*}
\end{prop}

We begin by proving a Lemma:
\begin{lem}
Let $q(X,Y)$ be a degree-$d$ polynomial and let $X$, $Y$ and $Z$ be independent standard Gaussians.  Then
$$
\var_Y\left[ \E_X\left[q(X,Y)^2 \right]\right] \leq 2^{O(d)} \E_{X,Y,Z}\left[(q(X,Y)-q(X,Z))^2 \right] \E_{X,Y}\left[ q(X,Y)^2 \right].
$$
\end{lem}
\begin{proof}
Recall that if $A,B$ are i.i.d. random variables, then $\var[A] = \frac{1}{2}\E[(A-B)^2]$.  We have that
\begin{align*}
\var_Y & \left[ \E_X\left[q(X,Y)^2 \right]\right] = \frac{1}{2}\E_{Y,Z}\left[\left(\E_X[q(X,Y)^2] - E_X[q(X,Z)^2]\right)^2\right]\\
& \leq 2^{O(d)} E_{Y,Z}\left[\left|\E_X[q(X,Y)^2-q(X,Z)^2]\right|\right]^2\\
& \leq 2^{O(d)} E_{X,Y,Z}\left[ |q(X,Y)-q(X,Z)||q(X,Y) + q(X,Z)|\right]^2\\
& \leq 2^{O(d)} \E_{X,Y,Z}\left[(q(X,Y)-q(X,Z))^2 \right]\E_{X,Y,Z}\left[(q(X,Y)+q(X,Z))^2 \right]\\
& \leq  2^{O(d)} \E_{X,Y,Z}\left[(q(X,Y)-q(X,Z))^2 \right]\E_{X,Y}\left[q(X,Y)^2 \right].
\end{align*}
The second line above is due to Corollary \ref{GeneralHypercontractiveCor}.  The fourth line is due to Cauchy-Schwarz.
\end{proof}

\begin{proof}[Proof of Proposition \ref{VarianceBoundProp}]
For fixed $X$, consider the polynomial
\begin{align*}
& q(( Y_2,Z_2,\ldots,Y_{\ell+1},Z_{\ell+1},W_1,\ldots,W_m),Y)\\ & =: D_{Y_2,Z_2}^\theta\cdots D_{Y_{\ell+1},Z_{\ell+1}}^\theta p(N^\theta_{W_1}\cdots N^\theta_{W_m}N^\theta_Y(X)).
\end{align*}
Let $V=(Y_2,Z_2,\ldots,Y_{\ell+1},Z_{\ell+1},W_1,\ldots,W_m)$.  Notice that
\begin{align*}
|p^{(\ell),m}_\theta(N^\theta_Y X)|_2^2 & = \E_V\left[ q(V,Y)^2 \right],\\
\theta|p^{(\ell+1),m}(X)|_2^2 & = \E_{V,Y,Z}\left[(q(V,Y)-q(V,Z))^2 \right],\\
|p^{(\ell),m+1}_\theta(X)|_2^2 & = \E_{V,Y}\left[ q(V,Y)^2 \right].
\end{align*}
Our result follows immediately upon applying the above Lemma to $q(V,Y)$.
\end{proof}

We also prove a relation between the higher averages
\begin{lem}\label{relationLem}
Let $d$ be an integer and $\theta=O(d^{-1})$ a real number.  Then there exist constants $c_0,\ldots,c_{2d+1}$ with $|c_m|=2^{O(d)}$ and $\sum_{m=0}^{2d+1} c_m = 0$ so that for any degree-$d$ polynomial $p$ and any vector $X$,
$$
\sum_{m=0}^{2d+1} c_m (A^\theta)^m p(X) = 0.
$$
\end{lem}
\begin{proof}
First note that $(A^\theta)^mp(X) = A^{\theta_m} p(X)$ where $\cos(\theta_m) = \cos(\theta)^m$.  We note that it suffices to find such $c$'s so that for any $Y$
$$
\sum_{m=0}^{2d+1} c_m p(N^{\theta_m}_Y(X)) = 0.
$$
Note that $\cos^2(\theta_m)=1-m\theta^2 + O(d\theta^4).$  Hence $\sin(\theta_m) = \sqrt{m}\theta +O(d\theta^3)$.

Recall that once $X$ and $Y$ are fixed, there is a degree-$2d$ polynomial $q$ so that for $z_m = e^{i\theta_m} = 1+i\sqrt{m}\theta +O(d\theta^2)$, $p(N^{\theta_m}_Y(X)) = z_m^{-d} q(z_m)$.  Hence we just need to pick $c_m$ so that for any degree-$2d$ polynomial $q$ we have that
$$
\sum_{m=0}^{2d+1} c_m z_m^{-d} q(z_m) =0.
$$
Such $c$ exist by standard interpolation results. In particular, it is sufficient to pick
$$
c_m = z_m^d \sqrt{(2d)!} \prod_{i=1,i\neq m}^{2d+1} \frac{\theta}{z_i-z_m}.
$$
For this choice of $c$ we have that
$$
|c_m| = \sqrt{(2d)!} \prod_{i=1,i\neq m}^{2d+1} \frac{\theta}{|z_i-z_m|} = \sqrt{(2d)!} \prod_{i=1,i\neq m}^{2d+1} \Theta\left(\frac{1}{|\sqrt{i}-\sqrt{m}|}\right)
$$
For $1\leq i \leq 2m$ we have that $|\sqrt{i}-\sqrt{m}| = \Theta(|i-m|/\sqrt{m})$.  For $i\geq 2m$, we have that $|\sqrt{i}-\sqrt{m}| = \Theta(\sqrt{i})$.  We evaluate $|c_m|$ based upon whether or not $2m \geq 2d+1$.  If $2m\geq 2d+1$, then
\begin{align*}
|c_m| & =  \sqrt{(2d)!} \prod_{i=1,i\neq m}^{2d+1} \Theta\left(\frac{\sqrt{m}}{|i-m|}\right) \\
& = 2^{O(d)} \sqrt{(2d)!} \frac{m^{d}}{(m-1)!(2d+1-m)!}\\
& = 2^{O(d)} \sqrt{d^{2d}} d^d \frac{\binom{2d}{m-1}}{(2d!)} \\
& = 2^{O(d)}.
\end{align*}
If $2m< 2d+1$,
\begin{align*}
|c_m| & =  \sqrt{(2d)!} \prod_{i=1,i\neq m}^{2m} \Theta\left(\frac{\sqrt{m}}{|i-m|}\right) \prod_{i=2m+1}^{2d+1} \Theta\left(\frac{1}{\sqrt{i}}\right)\\
& = 2^{O(d)} \sqrt{(2d)!} \frac{m^{m}}{(m-1)!(m-1)!}\sqrt{\frac{(2m)!}{(2d+1)!}}\\
& = 2^{O(d)} \frac{m^m \sqrt{(2m)!}}{(m-1)!(m-1)!} \\
& = 2^{O(d)} \frac{m^m \sqrt{(2m)!}\binom{2m-2}{m-1}}{(2m-2)!}\\
&= 2^{O(d)}.
\end{align*}

Considering $q(X)=1$, we find that $\sum_{m=0}^{2d+1} c_m=0$.
\end{proof}

Applying this to the polynomial $|p^{(\ell)}_\theta(X)|$, we find that
\begin{cor}\label{relationCor}
Let $d$ be an integer and $\theta=O(d^{-1})$ a sufficiently small real number (as a function of $d$).  Then there exist constants $c_0,\ldots,c_{4d+1}$ with $|c_m|=\Theta_d(1)$ and $\sum_{m=0}^{4d+1} c_m = 0$ so that for any degree-$d$ polynomial $p$, any vector $X$, and any non-negative integer $\ell$,
$$
\sum_{m=0}^{4d+1} c_m |p^{(\ell),m}_\theta (X)|_2^2 = 0.
$$
Furthermore $\sum_{m=0}^{4d+1} c_m = 0$ and $|c_m| = 2^{O(d)}$.
\end{cor}

In particular,
\begin{cor}\label{allCloseCor}
There exists some absolute constant $\alpha$, so that if $\theta=O(d^{-1})$ and if
$$\left|\log\left( \frac{|p^{(\ell),m}_\theta (X)|_2^2}{|p^{(\ell),m+1}_\theta (X)|_2^2}\right)\right| < \alpha^d$$
for some $\ell$ and all $1\leq m \leq 4d$, then all of the $|p^{(\ell),m}_\theta (X)|_2^2$ for that $\ell$ and all $0\leq m \leq 4d+1$ are within constant multiples of each other.
\end{cor}
\begin{proof}
It is clear that $|p^{(\ell),m}_\theta (X)|_2^2$ are within $1+O(d \alpha^d)$ of each other for $1\leq m\leq 4d+1$.  By Corollary \ref{relationCor}, we have that
\begin{align*}
|p^{(\ell),0}_\theta (X)|_2^2 &= \sum_{m=1}^{4d} \frac{-c_m}{c_0} |p^{(\ell),m}_\theta (X)|_2^2\\
& = \sum_{m=1}^{4d} \frac{-c_m}{c_0} |p^{(\ell),1}_\theta (X)|_2^2\left( 1 + O(d\alpha^d)\right)\\
& = \sum_{m=1}^{4d} \frac{-c_m}{c_0} |p^{(\ell),1}_\theta (X)|_2^2 + |p^{(\ell),1}_\theta (X)|_2^2 O\left(d\alpha^d \sum_{m=1}^{4d} \left| \frac{c_m}{c_0}\right|\right)\\
& = |p^{(\ell),1}_\theta (X)|_2^2\left(1 + O\left( d^2\alpha^d 2^{O(d)} \right) \right).
\end{align*}
Hence for $\alpha$ sufficiently small, $|p^{(\ell),0}_\theta (X)|_2^2$ is within a constant multiple of $|p^{(\ell),1}_\theta (X)|_2^2$.
\end{proof}

\section{Proof of Theorem \ref{MainThm}}\label{mainSec}

We now fix $c,\epsilon,d,N,k,p$ as in Theorem \ref{MainThm}.  Namely $c,\epsilon>0$, $d$ is a positive integer, $N$ an integer bigger than $B(c)^d \epsilon^{-4-c}$, and $k$ an integer bigger than $B(c)d$ for $B(c)$ some sufficiently large number depending only on $c$, and $p$ a degree-$d$ polynomial.  We fix $\theta=\arcsin\left(\frac{1}{\sqrt{N}}\right) \sim B(c)^{d/2} \epsilon^{-2-c/2}$.

Let $\rho:\R\rightarrow [0,1]$ be a smooth function so that $\rho(x)=0$ if $x<-1$ and $\rho(x)=1$ if $x>0$.  Let $\sigma:\R\rightarrow [0,1]$ a smooth function so that $\sigma(x) = 1$ if $|x|<1/3$ and $\sigma(x) = 0$ if $|x|>1/2$.  Let $\alpha$ be the constant given in Corollary \ref{allCloseCor}.

For $0\leq \ell \leq d$, $0\leq m \leq 4d+1$, let $q_{\ell,m}(X)$ be the degree-$2d$ polynomial $|p^{(\ell),m}_\theta(X)|_2^2$.  Recall by Lemma \ref{HighDerConstLem} that $q_{d,m}$ is constant.  We let $g_\pm(X)$ be
$$
I_{(0,\infty)}(\pm p(X))\prod_{\ell=0}^{d-1} \rho\left(\log\left(\frac{q_{\ell,0}(X)}{\epsilon^2 q_{\ell+1,0}(X) } \right) \right)\prod_{m=0}^{4d-1}\sigma\left(\alpha^{-d}\log\left(\frac{q_{\ell,m}(X)}{q_{\ell,m+1}(X)}\right)\right).
$$
Where $I_{(0,\infty)}(x)$ above is the indicator function of the set $(0,\infty)$.  Namely it is $1$ for $x>0$ and $0$ otherwise.

$g_\pm$ approximates the indicator functions of the sets where $p(X)$ is positive or negative.  To make this intuitive statement useful we prove:
\begin{lem}\label{gApproxLem}
The following hold:
\begin{itemize}
\item $g_\pm:\R^n\rightarrow [0,1]$
\item $g_+(X)+g_-(X)\leq 1$ for all $X$
\item For $Y$ a standard Gaussian $\E_Y[1-g_+(Y)-g_-(Y)] \leq 2^{O(d)} \epsilon $.
\end{itemize}
\end{lem}
\begin{proof}
The first two statements follow immediately from the definition.  The third statement follows by noting that $g_+(Y)+g_-(Y)=1$ unless $q_{\ell,0}(Y)<\epsilon^2 q_{\ell+1,0}(X)$ for some $\ell$ or $|q_{\ell,m+1}(X)-q_{\ell,m}(X)|<\Omega(\alpha^d)|q_{\ell,m}(X)|$ for some $\ell,m$.  By Corollary \ref{AvDerIneqCor} and Lemma \ref{AverageCloseLem}, this happens with probability at most $2^{O(d)}\epsilon$.
\end{proof}

We also want to know that the derivatives of $g_\pm$ are relatively small.
\begin{lem}\label{gDerBoundLem}
Consider $g_\pm$ as a function of $q_{\ell,m}(X)$ (consider $\sgn(p(X))$ to be constant).  Then the $t^{th}$ partial derivative of $g_\pm$, $\frac{\partial^t}{\partial q_{\ell_1,m_1}\partial q_{\ell_2,m_2}\cdots \partial q_{\ell_t,m_t}}$ is at most $O_t(1) 2^{O(dt)}\prod_{j=1}^t \frac{1}{q_{\ell_j,m_j}}$.
\end{lem}
\begin{proof}
The bound follows easily after considering $g$ as a function of the $\log(q_{\ell,m})$.  Noting that the $t^{th}$ partial derivatives of $q$ in terms of these logs are at most $O_t(\alpha^{-dt})$, the result follows easily.
\end{proof}

\begin{lem}\label{SmallChangeLem}
Given $c\leq 4$,
let $X$ be any vector and let $Y$ and $Z$ be $k$-independent families of Gaussians with $k\geq 512 c^{-1}d$, $N=\epsilon^{-4-c}$ and $\theta=O(d^{-2})$.  Then
$$
\left| \E_Y \left[g_+(N^\theta_Y(X))\right] - \E_Z\left[g_+(N^\theta_Z(X))\right] \right| \leq 2^{O_c(d)}\epsilon N^{-1}.
$$
And the analogous statement holds for $g_-$.
\end{lem}
\begin{proof}
First note that $\epsilon^2 \theta = O(\epsilon^{c/2})$, and hence that $(\epsilon^2 \theta)^{16/c}=O(\epsilon N^{-1}).$  Let $T$ be an even integer between $32/c$ and $64/c$.

First we deal with the case where $|\log(q_{\ell,m}(X)/q_{\ell,m+1}(X))| > \alpha^d$ for some $0\leq \ell\leq d, 1\leq m \leq 4d$, or $q_{\ell,1}(X) < \epsilon^2/10 q_{\ell+1,1}(X)$ for some $\ell$.  We claim that in either case the $\E\left[g_+(N^\theta_Y(X))\right] = O_d(\epsilon N^{-1})$ and a similar bound holds for $Z$.  If there is such an occurrence, find one with the largest possible $\ell$ and of the second type if possible for the same value of $\ell$.

Suppose that we had an occurrence of the first type.  Namely that for some $\ell,m$, $$|\log(q_{\ell,m}(X)/q_{\ell,m+1}(X))| > \alpha^d.$$   Pick such a one with $\ell$ maximal, and with $m$ minimal for this value of $\ell$.  Consider then the random variables $q_{\ell,m-1}(N^\theta_Y(X))$ and $q_{\ell,m}(N^\theta_Y(X))$.  They have means $q_{\ell,m}(X)$ and $q_{\ell,m+1}(X)$, respectively.  By Proposition \ref{VarianceBoundProp}, their variances are bounded by $$2^{O(d)}\theta q_{\ell+1,m-1}(X) q_{\ell,m}(X) \textrm{,  and  } 2^{O(d)}\theta q_{\ell+1,m}(X) q_{\ell,m+1}(X),$$ respectively.  Since we chose the smallest such $\ell$, all of the $q_{\ell+1,m}(X)$ for $1\leq m \leq 4d+1$ are close to $q_{\ell+1,1}(X)$ with multiplicative error at most $\alpha^d$.  By Corollary \ref{allCloseCor}, this implies that $q_{\ell+1,0}(X)$ is also close.  Hence, the variances of $q_{\ell,m-1}(N^\theta_Y(X))$ and $q_{\ell,m}(N^\theta_Y(X))$ are $O_d(\theta q_{\ell+1,1}(X) q_{\ell,m}(X))$ and $O_d(\theta q_{\ell+1,1}(X) q_{\ell,m+1}(X))$. Since there was no smaller $m$ to choose, $q_{\ell,m}(X)$ is within a constant multiple of $q_{\ell,1}(X)$.  Since we could not have picked an occurrence of the second type with the same $\ell$, we have that $q_{\ell,1}(X) \geq \epsilon^2 q_{\ell+1,1}(X)/10$.  Hence both of these variances are at most $$2^{O(d)}\epsilon^{-2} \theta \max(q_{\ell,m}(X),q_{\ell,m+1}(X))^2).$$  Hence by Corollary \ref{GeneralHypercontractiveCor}, for either of the random variables $Q_1=q_{\ell,m}(N^\theta_Y(X)),$ or $Q_2=q_{\ell,m+1}(N^\theta_Y(X))$ with means $\mu_1,\mu_2$ the $T^{th}$ moment of $|Q_i-\mu_i|$ (using the fact that $Y$ is at least $4Td$-independent) is at most $2^{O_c(d)}\epsilon N^{-1} \max(\mu_i)^T$.  Hence, with probability at least $1-2^{O_c(d)}\epsilon N^{-1}$, $|Q_i-\mu_i| < \alpha^d \max(\mu_i)/10$. But if this is the case, then $|\log(Q_1/Q_2)|$ will be more than $\alpha^d/2$, and $g_+$ will be 0.

Suppose that we had an occurrence of the second type for some $\ell$.  Again by Corollary \ref{allCloseCor}, we have that $q_{\ell+1,0}(X)$ is within a constant multiple of $q_{\ell+1,1}(X)$ and $q_{\ell+2,0}(X)$ within a constant multiple of $q_{\ell+2,1}(X)$.  Let $Q_0$ be the random variable $q_{\ell,0}(N^\theta(Y))$ and $Q_1$ the variable $q_{\ell+1,0}(N^\theta(Y))$.  We note that they have means equal to $q_{\ell,1}(X)$ and $q_{\ell+1,1}(X)$, respectively.  By Proposition \ref{VarianceBoundProp}, their variances are at most $2^{O(d)}\theta q_{\ell+1,1}(X)q_{\ell,1}(X)$ and $2^{O(d)}\theta q_{\ell+2,1}(X)q_{\ell+1,1}(X)$.  Since we had an occurrence at this $\ell$ but not the larger one, these are at most $2^{O(d)}\theta \epsilon^2 q_{\ell+1,1}(X)^2$ and $2^{O(d)}\theta \epsilon^{-2} q_{\ell+1,1}(X)^2$, respectively.  Considering the $T^{th}$ moment of $Q_1$ minus its mean, $\mu_1$, we find that with probability at least $1-2^{O_c(d)}\epsilon N^{-1}$ that $|Q_1-q_{\ell+1,1}|<q_{\ell+1,1}/20$.  Considering the $T^{th}$ moment of $Q_0$ minus its mean, we find that with probability at least $1-2^{O_c(d)}\epsilon N^{-1}$ that $|Q_0-q_{\ell,1}| < \epsilon^2 q_{\ell+1,1}/20$.  Together these imply that $\log(Q_0/(\epsilon^2 Q_1))<-1$ and hence that $g_+(N^\theta_Y(X))=0$.

Finally, we assume that neither of these cases occur.  We note by Corollary \ref{allCloseCor} that for each $m$ and $\ell$ that $q_{\ell,m}(X)$ is within a constant multiple of $q_{\ell,1}(X)$.  We define $Q_{\ell,m} = q_{\ell,m}(N^\theta_Y(X))$ and note by Proposition \ref{VarianceBoundProp} that $\var(Q_{\ell,m})$ is at most $2^{O(d)}\theta \epsilon^{-2} \E[Q_{\ell,m}]^2$.  We wish to show that this along with the $k$-independence of $Y$ is enough to determine $\E_Y[g_+(Q_{\ell,m})]$ to within $2^{O_c(d)}\epsilon N^{-1}$.  We do this by approximating $g_+$ by its Taylor series to degree $T-1$ about $(\E[Q_{\ell,m}])$.  The expectation of the Taylor polynomial is determined by the $4Td$-independence of $Y$.  We have left to show that the expectation of the Taylor error is small.  We split this error into cases based on whether $Q_{\ell,m}$ differs from its mean value by more than a constant multiple.

If no $Q_{\ell,m}$ varies by this much, the Taylor error is at most the sum over sequences $\ell_1,\ldots,\ell_{T},m_1,\ldots,m_{T}$ of $\prod_{i=1}^{T} |Q_{\ell_i,m_i} - \E[Q_{\ell_i,m_i}]|$ times an appropriate partial derivative.  Note that by Lemma \ref{gDerBoundLem} this derivative has size at most $O_d\left(\prod_{i=1}^{T} \frac{1}{|\E[Q_{\ell_i,m_i}]|}\right)$.  Noting that there are at most $2^{O_c(d)}$ such terms and that we can bound the expectation above as
$$\sum_{i=1}^T\left(\frac{|Q_{\ell_i,m_i}-\E[Q_{\ell_i,m_i}]|}{|\E[Q_{\ell_i,m_i}]|} \right)^T. $$
Since $T$ is even, the above is a polynomial in $Y$ of degree $4Td$.  Since $Y$ is $4Td$-independent, the expectation of the above is the same as it would be for $Y$ fully independent.  By Corollary \ref{GeneralHypercontractiveCor} this is at most
$$
2^{O_c(d)}\sum_{i=1}^T \left(\frac{\var_Y[Q_{\ell_i,m_i}]}{\E[Q_{\ell_i,m_i}]^2} \right)^{T/2}  \leq 2^{O_c(d)} (\theta \epsilon^{-2})^{T/2} \leq 2^{O_c(d)}\epsilon N^{-1}.
$$

If some $Q_{\ell,m}$ differs from its mean by a factor of more than 2, the Taylor error is at most 1 plus the size of our original Taylor term.  By Cauchy-Schwarz, the contribution to the error is at most the square root of the expectation of the square of the error term times the square root of the probability that one of the $Q_{\ell,m}$ varies by too much.  By an argument similar to the above, the former is $2^{O_c(d)}$.  To bound the latter, we consider the probability that a particular $Q_{\ell,m}$ varies by too much.  For this to happen $Q_{\ell,m}$ would need to differ from its mean by $2^{O(d)}\left(\theta\epsilon^{-2}\right)^{-1/2}$ times its standard deviation.  Using Corollary \ref{GeneralHypercontractiveCor} and the $8Td$-independence of $Y$, we bound this by considering the $2T^{th}$ moment of $Q_{\ell,m}$ minus its mean value, and obtain a probability of $2^{O_c(d)}(\epsilon N^{-1})^2$.  Hence this term produces a total error of $2^{O_c(d)}\epsilon N^{-1}$.

The argument for $g_-$ is analogous.
\end{proof}

\begin{cor}\label{FoolgCor}
If $\epsilon,c>0$, $X$ is the random variable described in Theorem \ref{MainThm} with $N\geq \epsilon^{-4-c}$, $N=O(d^{-2})$, $k\geq 512c^{-1} d$ and $Y$ is a fully independent family of random Gaussians then
$$
|\E[g_+(X)] - \E[g_+(Y)]| \leq 2^{O_c(d)}\epsilon.
$$
The same also holds for $g_-$.
\end{cor}
\begin{proof}
We let $Y_i$ be independent random standard Gaussians and let $Y=\frac{1}{\sqrt{N}}\sum_{i=1}^N Y_i$.  We let $Z^j = \frac{1}{\sqrt{N}} \left(\sum_{i=1}^j Y_i + \sum_{i=j+1}^N X_i\right)$.  Note that $Z^N = Y$ and $Z^0=X$.  We claim that
$$
|\E[g_+(Z^j)] - \E[g_+(Z^{j+1})]| = 2^{O_c(d)}\epsilon N^{-1},
$$
from which our result would follow.  Let $Z_j := \frac{1}{\sqrt{N-1}}\left( \sum_{i=1}^j Y_i  + \sum_{i=j+2}^N X_i\right)$, then this is
\begin{align*}
|\E_{Z_j,Y_j}[g_+(N^\theta_{Y_j}(Z_j))] & - \E_{Z_j,X_j}[g_+(N^\theta_{X_j}(Z_j))]|  \\ & \leq E_{Z_j}\left[ |\E_{Y_j}[g_+(N^\theta_{Y_j}(Z_j))] - \E_{X_j}[g_+(N^\theta_{X_j}(Z_j))]|\right],
\end{align*}
which by Lemma \ref{SmallChangeLem} is at most $2^{O_c(d)}\epsilon N^{-1}$.
\end{proof}

We can now prove Theorem \ref{MainThm}.
\begin{proof}
We prove that for $X$ as given with $N\geq\epsilon^{-4-c}$ and $k\geq 512 c^{-1}d$, and $Y$ fully independent that
$$
|\E[f(X)] - \E[f(Y)]| \leq 2^{O_c(d)}\epsilon.
$$
The rest will follow by replacing $\epsilon$ by $\epsilon' = \epsilon/2^{O_c(d)}$.  We note that $f$ is sandwiched between $2g_+-1$ and $1-2g_-$.  Now $X$ fools both of these functions to within $2^{O_c(d)}\epsilon$.  Furthermore by Lemma \ref{gApproxLem}, they have expectations that differ by $2^{O(d)}\epsilon$.  Therefore
$$
\E[f(X)] \leq \E[1-2g_-(X)] \approx_{2^{O_c(d)}\epsilon} \E[1-2g_-(Y)] \approx_{2^{O(d)}\epsilon} \E[f(Y)].
$$
We also have a similar lower bound.  This proves the Theorem.
\end{proof}

\section{Finite Entropy Version}\label{FiniteEntropySection}

The random variable $X$ described in the previous sections, although it does fool PTFs, has infinite entropy and hence cannot be used directly to make a PRG.  We fix this by instead using a finite entropy random variable that approximates $X$.  In order to make this work, we will need the following Lemma.

\begin{lem}\label{SmallErrorFoolLem}
Let $X_{i}$ be a $k$-independent family of Gaussians for $1\leq i\leq N$, so that the $X_i$ are independent of each other and $k,N$ satisfy the hypothesis in Theorem \ref{MainThm} for some $c,\epsilon,d$.  Let $\delta>0$.  Suppose that $Z_{i,j}$ are any random variables so that for each $i,j$, $|X_{i,j}-Z_{i,j}|<\delta$ with probability $1-\delta$.  Then the family of random variables $Z_j = \frac{1}{\sqrt{N}} \sum_{i=1}^N Z_{i,j}$ fools degree $d$ polynomial threshold functions up to $\epsilon + O(nN\delta + d\sqrt{nN}\log(\delta^{-1})\delta^{1/d}).$
\end{lem}

The basic idea is that with probability $1-nN\delta$, we will have that $|X_{i,j}-Z_{i,j}|<\delta$ for all $i,j$.  If that is the case, then for any polynomial $p$ it should be the case that $p(X)$ is close to $p(Z)$.  In particular, we show:

\begin{lem}\label{smallTildeChangeLem}
Let $p(X)$ be a polynomial of degree $d$ in $n$ variables.  Let $X\in \R^n$ be a vector with $|X|_\infty \leq B$ $(B>1)$.  Let $X'$ be another vector, so that $|X-X'|_\infty < \delta<1$.  Then
$$
|p(X)-p(X')| \leq \delta |p|_2 n^{d/2} O(B)^d
$$
\end{lem}
\begin{proof}
We begin by writing $p$ in terms of Hermite polynomials.  We can write
$$
p(X) = \sum_{i\in S} a_i h_i(X).
$$
Here $S$ is a set of size less than $n^d$, $h_i(X)$ is a Hermite polynomial of degree $d$ and $\sum_{i\in S} a_i^2 = |p|_2^2.$  The Hermite polynomial $h_i$ has $2^{O(d)}$ coefficients each of size $2^{O(d)}$.  Hence any of its partial derivatives at a point of $L^\infty$ norm at most $B$ is at most $O(B)^d$.  Hence by the intermediate value theorem, $|h_i(X)-h_i(X')| = \delta O(B)^d$.  Hence $|p(X)-p(X')| \leq \delta\sum_{i\in S} |a_i| O(B)^d.$  But $\sum_{i\in S} |a_i| \leq |p|_2 \sqrt{|S|}$ by Cauchy-Schwarz.  Hence
$$
|p(X)-p(X')| \leq \delta |p|_2 n^{d/2} O(B)^d.
$$
\end{proof}

We also need to know that it is unlikely that changing the value of $p$ by a little will change its sign.  In particular we have the following anticoncentration result, which is an easy consequence of \cite{anticoncentration} Theorem 8:

\begin{lem}[Carbery and Wright]\label{anticoncentrationLem}
If $p$ is a degree-$d$ polynomial then
$$
\pr(|p(X)| \leq \epsilon|p|_2) = O(d\epsilon^{1/d}).
$$
Where the probability is over $X$, a standard $n$-dimensional Gaussian.
\end{lem}

We are now ready to prove Lemma \ref{SmallErrorFoolLem}.

\begin{proof}
Note that with probability $1-\delta$ that $|X_{i,j}| = O(\log \delta^{-1})$.  Hence with probability $1-O(nN\delta)$ we have that $|X_{i,j}|_\infty = O(\log \delta^{-1})$ and $|X_{i,j}-Z_{i,j}|_\infty<\delta$.  Let $p$ be a degree $d$ polynomial normalized so that $|p|_2=1$.  We may think of $p$ as a function of $nN$ variables rather than just $N$, by thinking of $p(X)$ instead as $p\left( \frac{1}{\sqrt{N}}\sum_{i=1}^N X_i\right)$.  Applying Lemma \ref{smallTildeChangeLem}, we have therefore that with probability $1-O(nN\delta)$ that $|p(X)-p(Z)| < \delta O(\sqrt{nN}\log (\delta^{-1}))^d.$

We therefore have that if $Y$ is a standard family of Gaussians that
\begin{align*}
\pr(p(Z) < 0) & \leq O(nN\delta) + \pr(p(X) < \delta O(\sqrt{nN}\log (\delta^{-1}))^d)\\
&\leq \epsilon + O(nN\delta) + \pr(p(Y) < \delta O(\sqrt{nN}\log (\delta^{-1}))^d)\\
&\leq \epsilon + O(nN\delta + d\sqrt{nN}\log(\delta^{-1})\delta^{1/d}) + \pr(p(Y) < 0).
\end{align*}
The last step above following from Lemma \ref{anticoncentrationLem}.  We similarly get a bound in the other direction, completing the proof.
\end{proof}

We are now prepared to prove Corollary \ref{PRGCor}.

\begin{proof}
Given $\epsilon,c>0$, let $k,N$ be as required in the statement of Theorem \ref{MainThm}.  We will attempt to produce an effectively computable family of random variables $Z_{i,j}$ so that for some $k$-independent families of Gaussians $X_{i}$ we have that $|X_{i,j}-Z_{i,j}|<\delta$ with probability $1-\delta$ for each $i,j$ and $\delta$ sufficiently small.  Our result will then follow from Lemma \ref{SmallErrorFoolLem}.

Firstly, it is clear that in order to do this we need to understand how to actually effectively compute Gaussian random variables.  Note that if $u$ and $v$ are independent uniform $[0,1]$ random variables, then $\sqrt{-2\log(u)}\cos(2\pi v)$ is a Gaussian.  Hence we can let our $X_{i,j}$ be given by $$X_{i,j}=\sqrt{-2\log(u_{i,j})}\cos(2\pi v_{i,j}),$$ where $u_{i}$ and $v_i$ are $k$-independent families of uniform $[0,1]$ random variables.  We let $u_{i,j}',v_{i,j}'$ be $M$-bit approximations to $u_{i,j},v_{i,j}$ (i.e. $u_{i,j}'$ is $u_{i,j}$ rounded up to the nearest multiple of $2^{-M}$, and similarly for $v_{i,j}'$), and let $Z_{i,j}=\sqrt{-2\log(u_{i,j}')}\cos(2\pi v_{i,j}')$.  Note that we can equivalently compute $Z_{i,j}$ be letting $u_{i}',v_{i}'$ be $k$-independent families of variables taken uniformly from $\{2^{-M},2\cdot 2^{-M},\ldots,1\}$.  Hence, the $Z_{i,j}$ are effectively computable from a random seed of size $O(kNM)$.

We now need to show that $|X_{i,j}-Z_{i,j}|$ is small with high probability.  Let $a(u,v)=\sqrt{-2\log(u)}\cos(2\pi v)$.  Note that for $u,v\in[0,1]$ that $|a'| = O(1+u^{-1} + (1-u)^{-1/2})$.  Therefore, (unless $u_{i,j}'=1$) we have that since $X_{i,j} = a(u_{i,j},v_{i,j})$ and $Z_{i,j} = a(u_{i,j}',v_{i,j}')$, and since $|u_{i,j}-u_{i,j}'|,|v_{i,j}-v_{i,j}'|\leq 2^{-M}$, we have that
$$
|X_{i,j} - Z_{i,j} | = O(2^{-M}(1+u^{-1} + (1-u)^{-1/2})).
$$
Now letting $\delta= \Omega(2^{-M/2})$, we have that $2^{-M}(1+u^{-1} + (1-u)^{-1/2}) <\delta$ with probability more than $1-\delta$.  Hence for such $\delta$, we can apply Lemma \ref{SmallErrorFoolLem} and find that $Z$ fools degree $d$ polynomial threshold functions to within $\epsilon + O(nN\delta + d\sqrt{nN}\log(\delta^{-1})\delta^{1/d}).$  If $\delta < \epsilon^{3d} (dnN)^{-3d}$, then this is $O(\epsilon)$ (since for $x>d^{3d}$, we have that $x\log^{-d}(x)>x^{1/3}$).  Hence with $k=\Omega_c(d),N=2^{\Omega_c(d)}\epsilon^{-4-c}$ and $M=\Omega_c(d\log(dn\epsilon^{-1}))$, this gives us a PRG that $\epsilon$-fools degree $d$ polynomial threshold functions and has seed length $O(kNM)$.  Changing $c$ by a bit to absorb the $\log \epsilon^{-1}$ into the $\epsilon^{-4-c}$, and absorbing the $d\log d$ into the $2^{O_c(d)}$, this seed length is $\log(n)2^{O_c(d)}\epsilon^{-4-c}$.

\end{proof}

\end{document}